\documentclass[envcountsame]{llncs}

\usepackage{amsmath}
\usepackage{amssymb}
\usepackage{graphicx}
\usepackage{url}
\usepackage{algpseudocode}
\usepackage{algorithm}
\usepackage{comment}
\urldef{\mailsa}\path|{aravind, subruk, cs14resch01002}@iith.ac.in|
\urldef{\mailsb}\path|juho.lauri@tut.fi|

\usepackage{xspace}
\usepackage{framed}

\newcommand{\ProblemFormat}[1]{{\sc #1}}
\newcommand{\ProblemName}[1]{\ProblemFormat{#1}\xspace}

\newcommand{\probHappyEdges}{\ProblemName{Maximum Happy Edges}}
\newcommand{\probWHappyEdges}{\ProblemName{Weighted Maximum Happy Edges}}
\newcommand{\probHappyVertices}{\ProblemName{Maximum Happy Vertices}}
\newcommand{\probWHappyVertices}{\ProblemName{Weighted Maximum Happy Vertices}}
\newcommand{\probkWMHV}{\ProblemName{Weighted $k$-MHV}}
\newcommand{\probkWMHE}{\ProblemName{Weighted $k$-MHE}}
\newcommand{\probkMHV}{\ProblemName{$k$-MHV}}
\newcommand{\probkMHE}{\ProblemName{$k$-MHE}}
\newcommand{\probMultiwayCut}{\ProblemName{Multiway Cut}}
\newcommand{\probMultiwayUncut}{\ProblemName{Multiway Uncut}}

\newcommand{\probMaxWeightedPartition}{\ProblemName{Max Weighted Partition}}

\newcommand{\eps}{\varepsilon}

\newcommand{\heading}[1]{\medskip\noindent{\bf #1.\ }}%

\let\doendproof\endproof
\renewcommand\endproof{~\hfill\qed\doendproof}

\newtheorem{rrule}{Rule}
\spnewtheorem*{theorem*}{Theorem}{\bfseries}{\rmfamily}

\usepackage{chngcntr}

\spnewtheorem{subclaim}{Claim}{\bfseries}{\itshape}
\counterwithin{subclaim}{theorem}

\spnewtheorem{conj}[theorem]{Conjecture}{\bfseries}{\itshape}
\spnewtheorem{prob}[theorem]{Problem}{\bfseries}{\itshape}

\newenvironment{subproof}[1][\proofname]{%
  \renewcommand\endproof{~\hfill\ensuremath{\blacksquare}\doendproof}%
  \begin{proof}%
}{%
  \end{proof}%
}

\usepackage{booktabs}
\usepackage{colortbl}
\usepackage{xcolor}

\newcommand{\SB}{\{\,} \newcommand{\SM}{\;{|}\;} \newcommand{\SE}{\,\}}
\DeclareMathOperator{\tw}{tw}

\colorlet{tableheadcolor}{gray!25}
\colorlet{tablerowcolor}{gray!20} 
\newcommand{\rowcol}{\rowcolor{tablerowcolor}} %

\usepackage{subfig}

\makeatletter
\g@addto@macro\@floatboxreset\centering
\makeatother

\newcount\bsubfloatcount
\newtoks\bsubfloattoks
\newdimen\bsubfloatht

\makeatletter
\newcommand{\bsubfloat}[2][]{%
  \sbox\z@{#2}%
  \ifdim\bsubfloatht<\ht\z@
    \bsubfloatht=\ht\z@
  \fi
  \advance\bsubfloatcount\@ne
  \@namedef{bsubfloat\romannumeral\bsubfloatcount}{%
    \subfloat[#1]{\vbox to\bsubfloatht{\hbox{#2}\vfill}}}%
}
\newcommand{\resetbsubfloat}{\bsubfloatcount\z@\bsubfloatht=\z@}
\makeatother

\begin{document}

\title{Algorithms and hardness results for happy coloring problems}

\author{N. R. Aravind\inst{1} \and Subrahmanyam Kalyanasundaram\inst{1} \and Anjeneya Swami Kare\inst{1}\thanks{
Faculty member of University of Hyderabad. This work is carried out as part of his PhD program at IIT Hyderabad.} \and Juho Lauri\inst{2}\thanks{Work partially supported by the Emil Aaltonen Foundation (J.L.).}}
\institute{Department of Computer Science and Engineering\\
Indian Institute of Technology\\ Hyderabad, India\\ \mailsa \and Tampere University of Technology, Finland\\ \mailsb}

\maketitle
\thispagestyle{plain}
\begin{abstract}
In a vertex-colored graph, an edge is \emph{happy} if its endpoints have the same color. Similarly, a vertex is \emph{happy} if all its incident edges are happy.
Motivated by the computation of homophily in social networks, we consider the algorithmic aspects of the following \probHappyEdges (\probkMHE) problem: given a partially $k$-colored graph $G$, find an extended full $k$-coloring of $G$ maximizing the number of happy edges.
When we want to maximize the number of happy vertices, the problem is known as \probHappyVertices (\probkMHV).
We further study the complexity of the problems and their weighted variants.
For instance, we prove that for every $k \geq 3$, both problems are NP-complete for bipartite graphs, and \probkMHV remains hard for split graphs.
In terms of exact algorithms, we show both problems can be solved in time $O^*(2^n)$, and give an even faster $O^*(1.89^n)$-time algorithm when $k=3$.
From a parameterized perspective, we give a linear vertex kernel for \probkWMHE, where edges are weighted and the goal is to obtain happy edges of at least a specified total weight.
Finally, we prove both problems are solvable in polynomial-time when the graph has bounded treewidth or bounded neighborhood diversity.
\end{abstract}

\section{Introduction}
Analyzing large networks is of fundamental importance for a constantly growing number of applications. In particular, how does one mine e.g., social networks to provide valuable insight? A basic observation concerning the structure of social networks is \emph{homophily}, that is, the principle that we tend to share characteristics with our friends. Intuitively, it seems believable our friends are similar to us in terms of their age, gender, interests, opinions, and so on. In fact, this observation is well-known in sociology (see e.g.,~\cite{homph,McPherson2001,Lazarsfeld1954}). For example, imagine a network of supporters in a country with a two-party system. In order to check whether there is homophily by political stance (i.e., a person tends to befriend a person with similar political beliefs), we could count the number of edges between two people of opposite beliefs. If there were no such edges, we would observe homophily in an extreme sense. It is characteristic of social networks that they evolve over time: links tend to be added between people that share some characteristic. But given a snapshot of the network, how extensively can homophily be present? For instance, how far can an extreme ideology spread among people some of whom are ``politically neutral''?

We abstract these questions regarding the computation of homophily as follows.
Consider a vertex-colored graph $G=(V,E)$.
We say an edge is \emph{happy} if its endpoints have the same color (otherwise, the edge is \emph{unhappy}).
Similarly, a vertex is \emph{happy} if it and all its neighbors have the same color (otherwise, the vertex is \emph{unhappy}).
Equivalently, a vertex is happy when all of its incident edges are happy.
Let $S \subseteq V$, and let $c : S \to [k]$ be a partial vertex-coloring of $G$.
A full coloring $c' : V \to [k]$ is an \emph{extended full coloring} of $c$ if $c(v) = c'(v)$ for all $v \in S$.
In this paper, we consider the following coloring problems.

\begin{framed}
\vspace*{-0.25cm}
\noindent \probWHappyEdges (\probkWMHE) \\
\textbf{Instance:} A graph $G=(V,E)$, a partial vertex-coloring $c : S \subseteq V \to [k]$, a weight function $w : E \to \mathbb{N}$, and an integer $\ell$. \\
\textbf{Question:} Is there an extended full coloring $c'$ of $c$ such that the sum of the weights of the happy edges is at least $\ell$?
\vspace*{-0.25cm}
\end{framed}

\begin{framed}
\vspace*{-0.25cm}
\noindent \probWHappyVertices (\probkWMHV) \\
\textbf{Instance:} A graph $G=(V,E)$, a partial vertex-coloring $c : S \subseteq V \to [k]$, a weight function $w : V \to \mathbb{N}$, and an integer $\ell$. \\
\textbf{Question:} Is there an extended full coloring $c'$ of $c$ such that the sum of the weights of the happy vertices is at least $\ell$?
\vspace*{-0.25cm}
\end{framed}
\noindent We also consider the unweighted versions of the problems obtained by letting the weight of each edge or vertex be one.
We refer to these problems as \probkMHE and \probkMHV, respectively.

\heading{Previous work}
Zhang and Li~\cite{mhve} proved that for every $k \geq 3$, the problems \probkMHE and \probkMHV are NP-complete.
However, when $k=2$, they gave algorithms running in time $O(\min\{n^{2/3}m, m^{3/2} \})$ and $O(mn^7 \log n)$ for 2-MHE and 2-MHV, respectively.
Towards this end, the authors used max-flow algorithms ($2$-MHE) and minimization of submodular functions ($2$-MHV).
Moreover, the authors presented approximation algorithms with approximation ratios 1/2 and $\max \{ 1/k, \Omega(\Delta^{-3}) \}$ for \probkMHE and \probkMHV, respectively, where $\Delta$ is the maximum degree of the graph.
Later on, Zhang, Jiang, and Li~\cite{mhveapr} gave improved algorithms with approximation ratios 0.8535 and $1/(\Delta+1)$ for \probkMHE and \probkMHV, respectively.
In~\cite{kmhvelintree}, a subset of the current authors proved that both problems are solvable in polynomial time for trees.

Perhaps not surprisingly, the happy coloring problems are tightly related to cut problems.
Indeed, the \probkMHE problem is a generalization of the following \probMultiwayUncut problem~\cite{mwuc2006}.
In this problem, we are given a weighted undirected graph $G=(V,E)$ and a terminal set $S = \{s_1, s_2, \ldots, s_k\} \subseteq V$.
The goal is to find a partition of $V$ into classes $V_1,\ldots,V_k$ such that each class contains exactly one terminal and the total weight of the edges not cut by the partition is maximized. We obtain the \probMultiwayUncut problem as a special case of \probkMHE, when each color is used to precolor exactly one vertex.
We also mention that the complement of the \probMultiwayUncut problem is the \probMultiwayCut problem that has been studied before (see e.g.,~\cite{mwc92,mwc91}).
There are known (parameterized) algorithms for the \probMultiwayCut problem with the size of the cut $\ell$ as the parameter.
In this regard, the fastest known algorithm runs in $O^*(1.84^\ell)$ time~\cite{parmwcut}.

\heading{Our results}
We further study the complexity of happy coloring problems including their unweighted variants.
Our paper is organized as follows:
\begin{itemize}
\item In Section~\ref{sec:hardness}, we focus on hardness results of the problems and their unweighted variants for special graph classes. We prove that for every $k \geq 3$, the problem \probkMHV is NP-complete for split graphs and bipartite graphs. This extends the hardness result of Zhang and Li~\cite{mhve} for general graphs. Similarly, we show that \probkMHE remains NP-complete for bipartite graphs, and that its weighted variant is hard for complete graphs.

\item In Section~\ref{sec:exact}, we consider exact exponential-time algorithms for the happy coloring problems. The naive brute force runs in $k^n n^{O(1)}$ time, but we show that for every $k \geq 3$, there is an algorithm running in time $O^*(2^n)$, where $n$ is the number of vertices in the input graph. Moreover, we prove that this is not optimal for every $k$ by giving an even faster $O^*(1.89^n)$-time algorithm for both 3-MHE and 3-MHV.

\item In Section~\ref{sec:kernel}, we show that despite its hardness on complete graphs, the problem \probkWMHE admits a small kernel of size $k+\ell$ on general graphs. The ingredients of the kernel are a polynomial-time algorithm for \probkWMHE when the uncolored vertices induce a forest combined with simple reduction rules.

\item Finally, in Section~\ref{sec:tw-and-nd}, we study the complexity of both problems for sparse graphs (i.e., small treewidth) and the simplest of dense graphs (i.e., small neighborhood diversity). When either parameter is bounded, we show both problems admit polynomial-time algorithms.
\end{itemize}

\section{Preliminaries}
\label{sec:prelims}
We use standard asymptotic notation.
Sometimes, we write $f(n) = O^*(g(n))$ if $f(n) = O(g(n) \text{poly}(n))$, where $\text{poly}(n)$ represents any polynomial in $n$.
For a positive integer $n$, we use $[n]$ to denote the set $\{1,2,\ldots,n\}$.

Unless specified otherwise, all graphs we consider are simple and undirected.
Let $G=(V,E)$ be an undirected graph.
To reduce clutter, we use the shorthand $uv$ to denote an edge $\{u,v\} \in E$.
Two vertices $u,v \in V$ are \emph{adjacent} or \emph{neighbors} if $uv \in E$.
The set of neighbors of $v$ is the \emph{open neighborhood} of $v$, denoted by $N(v)$.
The \emph{closed neighborhood} of $v$, denoted by $N[v]$, is defined as $N[v] = N(v) \cup \{v\}$.
More generally, for $V' \subseteq V$, we write $N(V')$ and $N[V']$ to denote the set of all neighbors in the open neighborhood and closed neighborhood of each vertex in $V'$, respectively.
For graph-theoretic concepts not defined here, we refer the reader to~\cite{Diestel2005}.

A \emph{parameterized problem} is a language $L \subseteq \Sigma^* \times \mathbb{N}$, where $\Sigma$ is a fixed, finite alphabet. For an instance $(x,k) \in \Sigma^* \times \mathbb{N}$, we call $k$ the \emph{parameter}. The parameterized problem $L$ is \emph{fixed-parameter tractable} (FPT) when there is an algorithm $\mathcal{A}$, a computable function $f: \mathbb{N} \to \mathbb{N}$, and a constant $c$ such that, given $(x,k) \in \Sigma^* \times \mathbb{N}$, the algorithm $\mathcal{A}$ correctly decides whether $(x,k) \in L$ in time bounded by $f(k) \cdot |x|^c$. An equivalent way of proving a problem is FPT is by constructing a \emph{kernel} for it. A kernel for a parameterized problem $(x,k)$ is a polynomial-time algorithm $\mathcal{B}$ that returns an equivalent instance $(x',k')$ of $L$ such that $|x'| \leq g(k)$, for some computable function $g : \mathbb{N} \to \mathbb{N}$. Here, we say two instances are \emph{equivalent} if the first is a YES-instance iff the second is a YES-instance. Given a parameterized problem, it is a natural to ask whether it admits a kernel, and moreover whether that kernel is small.\footnote{This is abuse of notation: as is commonly done, we call the output of the kernel also a kernel.} By small, we typically mean a polynomial kernel, or even a linear kernel (i.e., $g(k) = O(k)$).

Kernelization is often discovered through \emph{reduction rules}. A reduction rule is a polynomial-time transformation of an instance $(x,k)$ to another instance of the same problem $(x',k')$ such that $|x'| < |x|$ and $k' \leq k$. A reduction rule is \emph{safe} when the instances are equivalent.
For more on parameterized complexity, we refer the interested reader to~\cite{fpt-book}.

\section{Further hardness results for happy coloring}
\label{sec:hardness}
Zhang and Li~\cite{mhve} showed that both \probkMHE and \probkMHV are NP-complete for every $k \geq 3$.
However, it appears their constructions do not enforce any particular structure on the graphs that are output.
This raises an immediate question: does the presence of some specific yet allowing graph structure enable us to solve the problems in polynomial-time?
We answer this question in the negative even for highly structured yet rich graph classes.
Our NP-completeness results are summarized in the Table~\ref{tbl:hardness-summary}.
All of our proofs are by a polynomial-time reduction from \probkMHE, which is NP-complete for every $k \geq 3$ by the result of~\cite{mhve}.


\begin{table}[t]
\caption{Summary of our hardness results for the unweighted problem variants.}
\label{tbl:hardness-summary}
\centering
\begin{tabular}{lll}
\toprule
Graph class\phantom{x} & \probkMHE\phantom{x} & \probkMHV \\
\midrule
Bipartite	& NPC & NPC \\
\rowcol Complete	& P & P \\
Split		& ? & NPC  \\
\bottomrule
\end{tabular}
\vspace*{-0.5cm}
\end{table}

\subsection{Hardness of \probkMHV for special graph classes}
Whenever possible, it is beneficial to prove hardness results for \probkMHV, as it is a special case of the more general \probkWMHV.
We begin by showing hardness for \emph{split graphs}, that is, for graphs whose vertex set can be partitioned into a clique and an independent set.
Clearly, complete graphs are also split graphs, and before proceeding we make the following easy observation.

\begin{proposition}
\label{prop:kmhv-clique}
Any partial coloring $c$ of the complete graph $K_n$ for any $n \geq 1$ can be extended to a full coloring $c'$ making $n$ vertices happy iff $c$ uses at most one color. Consequently, the problem \probkMHV is solvable in polynomial time for complete graphs for every $k \geq 1$.
\end{proposition}
Let us then proceed with the hardness result for split graphs.
Afterwards, we will modify the construction slightly to obtain a similar result for bipartite graphs.
\begin{theorem}
\label{thm:mhv-hardness-split}
For every $k \geq 3$, the problem \probkMHV is NP-complete for split graphs.
\end{theorem}
\begin{proof}
Let $I = (G, c, \ell)$ be an instance of \probkMHE, and let us in polynomial time construct an instance $I' = (G', c', \ell)$ of \probkMHV.
We can safely (and crucially) assume at least two vertices of $G$ are precolored (in distinct colors), for otherwise the instance is trivial.
We construct the split graph $G' = (C \cup B, E' \cup E'')$, where
\begin{itemize}
\item $C = \{ v_x \mid x \in V(G) \}$,
\item $B = \{ v_e \mid e \in E(G) \}$,
\item $E' = \{ v_ev_x \mid e \text{ is incident to } x \text { in } G \}$, and
\item $E'' = \{ v_x v_{x'} \mid x,x' \in V(G) \}$.
\end{itemize}
That is, $C$ forms a clique and $B$ an independent set in $G'$, proving $G'$ is split.
In particular, observe that the degree of each vertex $v_e$ is two.
To complete the construction, we retain the precoloring, i.e., set $c'(v_x) = c(x)$ for every $x \in V(G)$.
The construction is illustrated in Figure~\ref{fig:mhv-hardness-split}.

We claim that $I$ is a YES-instance of \probkMHE iff $I'$ is a YES-instance of \probkMHV.
Suppose $\ell$ edges can be made happy in $G$ by an extended full coloring of $c$.
Consider an edge $e \in E(G)$ whose endpoints are colored with color $i$.
To make $\ell$ vertices happy in $G'$, we give $v_e$ and its two neighbors the color $i$.
For the other direction, suppose $\ell$ vertices are happy under an extended full coloring of $c'$.
As at least two vertices in $C$ are colored in distinct colors, it follows by Proposition~\ref{prop:kmhv-clique} that all the happy vertices must be in $B$.
Furthermore, the vertices in $B$ correspond to precisely the edges in $E(G)$, so we are done.
\end{proof}

\begin{figure}[t]
\bsubfloat[]{%
  \includegraphics[scale=1.25]{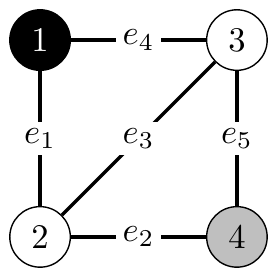}%
}
\bsubfloat[]{%
  \includegraphics[scale=1.25]{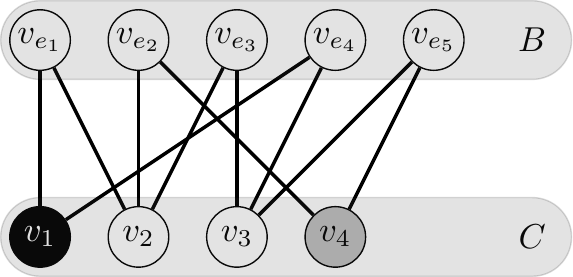}%
}
\bsubfloati\qquad\bsubfloatii
\caption{\textbf{(a)} A graph $G$ of an instance of \probkMHE, where white vertices correspond to uncolored vertices. \textbf{(b)} The graph $G$ transformed into a split graph $G'$ by the construction of Theorem~\ref{thm:mhv-hardness-split}. The edges between the vertices in $C$ are not drawn.}
\vspace{-0.5cm}
\label{fig:mhv-hardness-split}
\end{figure}

\begin{theorem}
\label{thm:mhv-hardness-bipartite}
For every $k \geq 3$, the problem \probkMHV is NP-complete for bipartite graphs.
\end{theorem}

\begin{proof}
We start with the construction of Theorem~\ref{thm:mhv-hardness-split}.
Modify the split graph $G'$ by deleting the edges between the vertices in $C$, i.e., let $G' = (C \cup B, E')$.
For each $v_x \in C$, add a path $S_{v_x} = \{ v_x^1, v_x^2, v_x^3 \}$ along with the edges $v_xv_x^1$ and $v_x^3v_x$.
In other words, each $v_x$ forms a 4-cycle with the vertices in $S_{v_x}$.
Clearly, we have that $G'$ is bipartite as it contains no odd cycles.
Arbitrarily choose three distinct colors from $[k]$, and map them bijectively to $S_{v_x}$.
Observe that by construction, none of the vertices in $S_{v_x}$ can be happy under any $c'$ extending $c$.
This completes the construction.
Correctness follows by the same argument as in Theorem~\ref{thm:mhv-hardness-split}.
\end{proof}

\subsection{Hardness of \probkMHE for special graph classes}
In the spirit of the previous subsection, let us begin with an observation regarding the polynomial-time solvability of the problem for complete graphs.

\begin{proposition}
\label{prop:kmhe-clique}
The problem \probkMHE is solvable in polynomial time for complete graphs for every $k \geq 1$.
\end{proposition}

\begin{proof}
Let $S$ denote the set of precolored vertices for the $K_n$ for any $n \geq 1$.
Delete edges whose both endpoints are in $S$, since their happiness is already determined by the precoloring.
Observe that $S$ is now an independent set and $C = V \setminus S$ induces a clique.
Moreover, every vertex in $S$ is adjacent to every vertex in $C$.

Denote by $p$ the most frequent occurrence of any color among the precolored vertices.
For any vertex $v \in C$, regardless of the color we give to $v$, we can make at most $p$ edges happy among the edges from the vertices in $S$ to $v$.
Thus, the number of happy edges is at most $p \cdot |C| + |E(C)|$.
In fact, we can achieve exactly $p \cdot |C| + |E(C)|$ happy edges by giving a single color to all the vertices in $C$.
More precisely, we color all the uncolored vertices with the color that is used $p$ times, completing the proof.
\end{proof}

We remark that for the above proof to hold, we do not need the graph to be complete.
Indeed, the procedure described in the proof can be applied as long as every precolored vertex is adjacent to every uncolored vertex.

We now turn to hardness results, and show that for every $k \geq 3$, the problem \probkMHE is NP-complete for bipartite graphs as well.
\begin{theorem}
\label{thm:mhe-hardness-bipartite}
For every $k \geq 3$, the problem \probkMHE is NP-complete for bipartite graphs.
\end{theorem}

\begin{proof}
Let $I = (G, c, \ell)$ be an instance of \probkMHE, and let us in polynomial time construct an instance $I' = (G', c, m+\ell)$ of \probkMHE, where $G'$ is bipartite.
We obtain $G'$ by subdividing every edge of $G$.
Observe that if $G$ has $n$ vertices and $m$ edges, then $G'$ has $n+m$ vertices and $2m$ edges.
Clearly, $G'$ is bipartite.

We will now show that $G$ has an extended full coloring making at least $\ell$ edges happy iff $G'$ has an extended full coloring making at least $m + \ell$ edges happy.
Let $c'$ be an extended full coloring of the precoloring $c$ given to $G$.
We give $G'$ the same extended full coloring, and give each vertex in $v \in V(G') \setminus V(G)$ an arbitrary color that appears on a vertex adjacent to $v$.
Thus, for each edge in $G$, we have one extra happy edge in $G'$, giving us a total of at least $m + \ell$ happy edges.
For the other direction, let $c'$ be an extended full coloring of $c$ that makes at least $m + \ell$ edges happy in $G'$.
Now, there are at least $\ell$ vertices in $V(G') \setminus V(G)$ with both of its incident edges happy.
These $2 \ell$ happy edges correspond to the $\ell$ happy edges in $G$.
This concludes the proof.
\end{proof}

We are unable to prove that \probkMHE remains NP-complete for split graphs.
However, when weights are introduced, it is easy to see that the problem remains NP-complete even for complete graphs for every $k \geq 3$.
\begin{theorem}
\label{thm:wmhe-hardness-complete}
The problem \probkWMHE is NP-complete for complete graphs for every $k \geq 3$.
\end{theorem}

\begin{proof}
Let $I = (G, c, \ell)$ be an instance of \probkMHE, and let us in polynomial time construct an instance $I' = (G', c', w, \alpha \cdot \ell)$ of \probkWMHE, where $G'$ is a complete graph.
Here $\alpha = {n \choose 2} - m + 1$, where $n = |V(G)|$ and $m = |E(G)|$.
To construct $G'$, first let $V(G') = V(G)$.
For each $uv \in E(G)$, insert the edge $uv$ to $G'$ with weight $\alpha$.
For each $uv \notin E(G)$, insert the edge $uv$ to $G'$ with weight $1$.
Clearly, $G'$ is a complete graph.
It is easy to see $G$ has at least $\ell$ happy edges iff $G'$ has happy edges of total weight at least $\alpha \cdot \ell$.
\end{proof}

\section{Exact exponential-time algorithms for happy coloring}
\label{sec:exact}
In this section, we consider the happy coloring problems from the viewpoint of exact exponential-time algorithms.
Every problem in NP can be solved in time exponential in the input size by a brute-force algorithm.
For \probkMHE (\probkMHV), such an algorithm goes through each of the at most $k^n$ colorings, and outputs the one maximizing the number of happy edges (vertices).
It is natural to ask whether there is an algorithm that is considerably faster than the $k^n n^{O(1)}$-time brute force approach.
In what follows, we show that brute-force can be beaten.

Let us introduce the following more general problem.
\begin{framed}
\vspace*{-0.25cm}
\noindent \probMaxWeightedPartition \\
\textbf{Instance:} An $n$-element set $N$, integer $d$, and functions $f_1,f_2,\ldots,f_d : 2^N \to [-M,M]$ for some integer $M$. \\
\textbf{Question:} A $d$-partition $(S_1,S_2,\ldots,S_d)$ of $N$ that maximizes $f_1(S_1) + f_2(S_2) + \cdots + f_d(S_d)$.
\vspace*{-0.25cm}
\end{framed}
\noindent Using an algebraic approach, the following has been shown regarding the complexity of the problem.
\begin{theorem}[Bj\"{o}rklund, Husfeldt, Koivisto~\cite{Koivisto2009}]
\label{thm:mwp}
The \probMaxWeightedPartition problem can be solved in
\begin{itemize}
\item $2^n d^2 M \cdot n^{O(1)}$ time and exponential space, and
\item $3^n d^2 M \cdot n^{O(1)}$ time and polynomial space.
\end{itemize}
\end{theorem}
In the following, we observe that the weighted variants of both problems can be reduced to \probMaxWeightedPartition. This results in an algorithm that is considerably faster than one running in time $k^n n^{O(1)}$.

\begin{lemma}
\label{lem:reduction-mwp}
For every $k \geq 1$, the problems \probkWMHE and \probkWMHV reduce in polynomial time to \probMaxWeightedPartition.
\end{lemma}

\begin{proof}
Consider the claim for an instance $I = (G,c,w,\ell)$ of \probkWMHE.
To construct an instance of \probMaxWeightedPartition, let $N = V(G) \setminus S$, where $S$ is the set of precolored vertices, let $d = k$, and let $M = \sum_{uv \in E(G)} w(uv)$.
Define $f_i = \sum_{uv \in E(G[ S_i \cup c^{-1}(i) ])} w(uv)$, i.e., $f_i$ sums the weights of the edges $uv$ that range over the edge set of the subgraph induced by the union of $S_i$ and $c^{-1}(i)$, the vertices precolored with color $i$. Thus, a partition $(S_1,\ldots,S_k)$ maximizing $f_1(S_1) + \cdots + f_k(S_k)$ maximizes the weight of happy edges.

Finally, consider the claim for an instance $I = (G,c,w,\ell)$ of \probkWMHV.
Now, we define $f_i = \sum_{v \in S_i : \forall y \in N(v) : y \in ( S_i \cup c^{-1}(i) ) } w(v)$, i.e., $f_i$ sums the weights of the vertices $v$ for which it holds that $v$ and each neighbor $y$ of $v$ are all colored with color $i$. Also, we let $M = \sum_{v \in V(G)} w(v)$, but otherwise the argument is the same as above.
\end{proof}

For some NP-complete problems, the fastest known algorithms run in $O^*(2^n)$ time, but we do not necessarily know whether (under reasonable complexity-theoretic assumptions) they are optimal.
Indeed, could one have an algorithm that runs in $O^*((2-\eps)^n)$ time, for any $\eps > 0$, for either \probkMHE or \probkMHV?
We prove that at least for some values of $k$ this bound can be achieved.
For this, we recall the following result.
\begin{theorem}[Zhang and Li~\cite{mhve}]
\label{thm:2mhe}
For $k=2$, the problems \probkMHE and \probkMHV can be solved in time $O(\min\{n^{2/3}m, m^{3/2} \})$ and $O(mn^7 \log n)$, respectively.
\end{theorem}
We are ready to proceed with the following.

\begin{lemma}
\label{lem:3mhe-exact}
For $k = 3$, the problems \probkMHE and \probkMHV can be solved in time $O^*(1.89^{n'})$, where $n'$ is the number of uncolored vertices in the input graph.
\end{lemma}
\begin{proof}
First, consider the claim for an instance $I = (G,c,w,\ell)$ of \probkMHE.
Consider a partition $\mathcal{S} = (S_1,S_2,S_3)$ of
the uncolored vertices
into $k=3$ color classes that maximizes the number of happy edges.
In $\mathcal{S} \setminus S_3$, by the optimality of $\mathcal{S}$, it must be the case $S_1$ and $S_2$ have a minimum number of crossing edges.
Thus, we can proceed as follows.
Observe that in any optimal solution $\mathcal{S}$, there exists $S_i \in \mathcal{S}$ such that $|S_i| \leq n'/3$.
The number of subsets of size at most $n'/3$ is $2^{H(1/3)n'} < 1.89^{n'}$, using the well-known bound $2^{H(1/3)} < 1.89$, where $H(\cdot)$ is the binary entropy function (for a proof, see e.g.,~\cite[Lemma~3.13]{Fomin2010}).
Thus, we guess $S_i$ by extending it in all possible at most $1.89^{n'}$ ways.
Then, for every such partial coloring, we solve an instance of 2-MHE on the remaining uncolored vertices in polynomial time by Theorem~\ref{thm:2mhe}.
Combining the bounds, we obtain an algorithm running in time $O^*(1.89^{n'})$ for 3-MHE.

The observation is similar for 3-MHV, but we solve an instance of 2-MHV on $V(G) \setminus N[S_i]$ instead of $V(G) \setminus S_i$.
\end{proof}
By Lemma~\ref{lem:3mhe-exact} and by combining Theorem~\ref{thm:mwp} with Lemma~\ref{lem:reduction-mwp}, we arrive at the following.
\begin{theorem}
\label{thm:exact-algs}
For every $k \geq 3$, the problems \probkMHE and \probkMHV can be solved in time $O^*(2^{n'})$. When $k = 3$, the problems are solvable in time $O^*(1.89^{n'})$, where $n'$ is the number of uncolored vertices in the input graph.
\end{theorem}

\section{A linear kernel for \probkWMHE}
\label{sec:kernel}
In this section, we prove that \probkWMHE has a kernel of size $k + \ell$.
Our strategy to obtain the kernel consists of two parts: first, we will show that there is a polynomial-time algorithm for the problem where the uncolored vertices induce a forest.
Then, to leverage this algorithm, we apply a set of reduction rules that shrink the instance considerably, or solve it directly along the way.

\subsection{A linear-time algorithm for subproblems of \probkWMHE}
\label{polycase}
We show that the \probkWMHE problem is polynomial-time solvable when the uncolored vertices $V \setminus S$ induce a tree, where $S$ is the set of precolored vertices.
When $V \setminus S$ induces a forest, we run the algorithm for each component in $V \setminus S$ independently.
The approach we present is based on dynamic programming, and inspired by the algorithm given in~\cite{kmhvelintree}.

We define edges \emph{touching} a subtree to be those edges that have at least one endpoint in the subtree.
We choose any vertex $r \in V \setminus S$ as the root of the tree induced by $V \setminus S$.
The vertices of this rooted tree are processed according to its post-order traversal.
At each node, we keep $k$ values.
The $k$ values are defined as follows,
for $1 \leq i \leq k$:

\begin{itemize}
  \item $T_{v}[i]:$ The maximum total weight of the happy edges touching the subtree $T_{v}$, when the vertex
  $v$ is colored with color $i$.
\end{itemize}
We also define the following expressions:
\begin{itemize}
  \item $T_{v}[*]:$ The maximum total weight of the happy edges touching the subtree $T_{v}$, i.e.,
  \begin{equation}
  T_{v}[*] = \max_{i=1}^{k}\{T_{v}[i]\}.
\end{equation}
  \item $T_{v}[\overline \imath]:$ The maximum total weight of the happy edges touching the subtree $T_{v}$, when the vertex
  $v$ is colored with a color other than $i$, i.e.,
  \begin{equation}
  T_{v}[\overline \imath] = \max_{j =1,j \neq i}^{k} \{T_{v}[j]\}.
\end{equation}
\end{itemize}

If $W_p$ is the total weight of the happy edges in the initial partial coloring, $W_p + T_{r}[*]$ gives us
the maximum total weight of the happy edges in $G$. Now, we explain how to compute the values $T_{v}[i]$ for $1 \leq i \leq k$ and for each $v \in V \setminus S$.
When we say \emph{color-$i$} vertices, we mean the vertices precolored with color $i$.

For a leaf vertex $v \in V \setminus S$, let $v_{1}, v_{2}, \ldots, v_{x}$ be the color-$i$ neighbors of $v$ in $G$.
Then,
\begin{equation}
  T_{v}[i] = \sum_{j = 1}^{x} w(v v_{j}).
\end{equation}
If there are no color-$i$ neighbors for $v$, then $T_{v}[i]$ is set to $0$.

For a non-leaf vertex $v \in V \setminus S$, let $v_{1}, v_{2}, \ldots, v_{x}$ be the color-$i$ neighbors of $v$ in $G$ and let $u_{1}, u_{2}, \ldots, u_{d}$ be the children of $v$ in $V \setminus S$.
Then,

\begin{equation}
\label{mhenleq}
  T_{v}[i] = \sum_{j = 1}^{x} w(v v_{j}) + \sum_{j=1}^{d} \max \{(w(v u_{j}) + T_{u_{j}}[i]),  T_{u_{j}}[\overline \imath]\}.
\end{equation}

This naturally leads to an algorithm listed as Algorithm~\ref{alg:mhe}.
The running time of the algorithm is $O(k (m + n))$.
The correctness of the values $T_{v}[i]$, for $1 \leq i \leq k$ and for each $v \in V \setminus S$, implies the correctness of the algorithm.
The following theorem is proved by induction on the size of the subtrees.

\begin{algorithm}[t]
\renewcommand{\algorithmicrequire}{\textbf{Input:}}
\renewcommand{\algorithmicensure}{\textbf{Output:}}
\caption{Algorithm for a special case of \probkWMHE}\label{alg:mhe}
\begin{algorithmic}[1]
\Require A weighted undirected graph $G$ with $S \subseteq V$ precolored vertices under a partial vertex-coloring $c : V \to [k]$,
$V \setminus S$ induces a tree, and a vertex $r \in V \setminus S$ as the root of the tree.
\Ensure Maximum total weight of the happy edges in $G$.
    \State $M_{p} \gets 0$
    \ForAll{happy edge $uv$ in the precoloring}
       \State $M_{p} \gets M_{p} + w(uv)$
    \EndFor
    \ForAll{$v \in V \setminus S$ in post-order}
        \If {$v$ is a leaf vertex in $V \setminus S$}
            \For{$i = 1 \mbox{ to } k$}
                \State $T_{v}[i] \gets 0$
                \ForAll{$vu \in E$ such that $u \in S$ and $c(u) = i$}
                    \State $T_{v}[i] \gets T_{v}[i] + w(vu)$
                \EndFor
            \EndFor
        \Else
            \For{$i = 1 \mbox{ to } k$}
                \State $T_{v}[i] \gets 0$
                \ForAll{$vu \in E$ such that $u \in S$ and $c(u) = i$}
                    \State $T_{v}[i] \gets T_{v}[i] + w(vu)$
                \EndFor
                \ForAll{child $u$ of $v$ in $V \setminus S$}
                    \State $T_{v}[i] \gets T_{v}[i] + \max\{(w(vu)+T_{u}[i]), T_{u}[\overline \imath] \}$
                \EndFor
            \EndFor
        \EndIf
    \EndFor
    \State \Return ($M_{p} + T_{r}[*]$)
\end{algorithmic}
\end{algorithm}

\begin{theorem}
\label{thm:alg-correctness}
Algorithm~\ref{alg:mhe} correctly computes the values $T_{v}[i]$ for every
$v \in V \setminus S$ and $1 \leq i \leq k$.
\end{theorem}

\begin{proof}
We prove the theorem by using induction on the size of the subtrees.
For a leaf vertex $v$, the algorithm correctly computes the values $T_{v}[i]$ for $1 \leq i \leq k$.
For a non-leaf vertex $v$, let $u_{1}, u_{2}, \ldots, u_{d}$ be the children of $v$ in $V \setminus S$.
By induction, all the $k$ values associated with each child $u_{j}$
of $v$ are correctly computed.
Moreover, $T_{v}[i]$ is the sum of two quantities (see Equation~\ref{mhenleq}),
the first quantity is correct because it is the sum of the weights of the happy edges from $v$ to $S$.
If $T_{v}[i]$ is not correct, it will contradict the correctness of $T_{u_{j}}[*]$ for some child
$u_{j}$ of $v$. So, the second term in the $T_{v}[i]$ is correct. Hence, the algorithm correctly
computes the values $T_{v}[i]$ for every $v$ in $V \setminus S$ and $1 \leq i \leq k$.
\end{proof}

\subsection{Reduction rules combined with the algorithm: a kernel}
\label{rules}
In this subsection, we assume the edge weights of the \probkWMHE instance are positive integers.
The kernel will also work for real weights that are at least 1.
We present the following simple reduction rules.

\begin{rrule}
\label{rule:isolated}
If $G$ contains an isolated vertex, delete it.
\end{rrule}

\begin{rrule}
\label{rule:edge}
If both endpoints of an edge $uv \in E$ are colored, remove $uv$. Furthermore, if $c(u) = c(v)$, decrement $\ell$ by the weight on $uv$.
\end{rrule}

\begin{proof}
As both endpoints of $uv$ are colored, the existence of the edge $uv$ does not further contribute to the value of the optimal solution.
Moreover, if the edge is already happy under $c$, we can safely decrement $\ell$.
\end{proof}

\begin{rrule}
\label{rule:merge}
Contract every color class $C_i$ induced by the partial coloring $c$ into a single vertex. Let $e_1,\ldots,e_r$ be the (parallel) edges between two vertices $u$ and $v$. Delete each edge in $e_1,\ldots,e_r$ except for $e_1$, and update $w(e_1) = w(e_1) + w(e_2) + \cdots + w(e_r)$.
\end{rrule}

\begin{proof}
Let $G'$ be the resulting graph after the application of Rule~\ref{rule:merge}.
Because Rule~\ref{rule:edge} does not apply, each color class $C_i$ forms an independent set.
Thus, $G'$ contains no self-loops.

Fix a color $i$, and consider an uncolored vertex $v \in V \setminus C_i$.
Denote by $N_i(v)$ the neighbors of $v$ with color $i$, and denote by $E[X,Y]$ the set of edges whose one endpoint is in $X$ and the other in $Y$.
Depending on the color $v$ gets in an extended full coloring of $c$, either all edges in $E[\{v\}, N_i(v)]$ are happy or all are unhappy.
Hence, we can safely replace these edges with a single weighted edge.
\end{proof}

\begin{theorem}
\label{thm:mhe-kernel}
The problem \probkWMHE admits a kernel on $k + \ell $ vertices.
\end{theorem}
\begin{proof}
Let $(G,c,w,\ell)$ be a reduced instance of \probkWMHE.
We claim that if $G$ has more than $k + \ell$ vertices, then we have YES-instance.
The proof follows by the claims below.

\begin{subclaim}
The weight of each edge is at most~$\ell$.
\end{subclaim}
\begin{subproof}
If an edge $uv$ has $w(uv) \geq \ell$ and at least one of $u$ and $v$ is uncolored, we make $uv$ happy and output YES.
On the other hand, any unhappy edge (with any weight) has been removed by Rule~\ref{rule:edge}.
\end{subproof}

\begin{subclaim}
The number of precolored vertices in $G$ is at most~$k$.
\end{subclaim}
\begin{subproof}
Follows directly from Rule~\ref{rule:merge}.
\end{subproof}

\begin{subclaim}
The number of uncolored vertices in $G$ is at most $\ell-1$.
\end{subclaim}
\begin{subproof}
Let $H$ be the graph induced by the uncolored vertices, i.e., $H = G[V \setminus \cup_{i \in [k]} C_i]$.
We note the following two cases:
\begin{itemize}
    \item If any of the connected components of $H$ is a tree, then we apply the procedure described in Sction~\ref{polycase} for that
component, and decrement the parameter $\ell$ accordingly.
    \item If $w(E(H)) \geq \ell$, then we color all the vertices in $H$ by the same color making all the
edges in $H$ happy. So the case where $w(E(H)) \geq \ell$ is a YES-instance.
\end{itemize}
After the application of the above, every component of $H$ contains a cycle, and $|E(H)| < \ell$.
So in each component of $H$, the number of vertices is at most
the number of edges. Consequently, we have $|V(H)| \leq |E(H)| < \ell$. Hence the number of uncolored vertices
is at most $\ell - 1$.

\end{subproof}

Clearly, all of the mentioned rules can be implemented to run in polynomial time.
Moreover, as we have bounded the number of precolored and uncolored vertices, the claimed kernel follows.
\end{proof}
By combining Theorem~\ref{thm:mhe-kernel} with Theorem~\ref{thm:exact-algs}, we have the following corollary.
\begin{corollary}
For every $k \geq 3$, \probkMHE can be solved in time $O^*(2^{\ell})$. For the special case of $k=3$, the problem admits an algorithm running in time $O^*(1.89^{\ell})$.
\end{corollary}

\section{Structural parameterization: density and sparsity}
\label{sec:tw-and-nd}
In this section we consider the happy coloring problems on sparse and dense graphs through a structural parameterization.

We begin by recalling a widely known measure for graph sparsity.
A \emph{tree decomposition} of $G$ is a pair
$(T,\{X_i : i\in I\})$
where $X_i \subseteq V$, $i\in I$, and $T$ is a tree with elements
of $I$ as nodes
such that:
\begin{enumerate}
\item for each edge $uv\in E$, there is an $i\in I$ such that $\{u,v\}
\subseteq X_i$, and
\item for each vertex $v\in V$, $T[\SB i\in I \SM v\in X_i \SE]$ is a (connected) tree with at least one node.
\end{enumerate}
The \emph{width} of a tree decomposition is $\max_{i \in I} |X_i|-1$.
The \emph{treewidth}~\cite{Robertson1986} of $G$
is the minimum width taken over all tree decompositions
of $G$ and it is denoted by $\tw(G)$.
For algorithmic purposes, it is convenient to consider the following \emph{nice tree decomposition} of a decomposition $(T,\{X_i : i\in I\})$ where every node $i \in I$ is one of the following types:
\begin{enumerate}
\item Leaf: node $i$ is a leaf of $T$ and $|X_i| = 1$.
\item Introduce: node $i$ has exactly one child $j$ and there is a vertex $v \in V$ with $X_i = X_j \cup \{ v \}$.
\item Forget: node $i$ has exactly one child $j$ and there is a vertex $v \in V$ with $X_j = X_i \cup \{ v \}$.
\item Join: node $i$ has exactly two children $j_1$ and $j_2$ and $X_i = X_{j_1} = X_{j_2}$.
\end{enumerate}
Every $n$-vertex graph $G$ has a nice tree decomposition with $O(n)$ nodes and width equal to $\tw(G)$.
Moreover, such a decomposition can be found in linear time if $\tw(G)$ is bounded (see e.g.,~\cite{Bodlaender1996}).

To show both weighted variants of happy coloring are tractable for bounded treewidth graphs, we proceed with a standard application of dynamic programming over a tree decomposition.
For more details, we refer the reader to~\cite{Bodlaender2008}.
\begin{theorem}
\label{thm:tw-alg}
For any $k \geq 1$, both \probkWMHE and \probkWMHV can be solved in time $k^t \cdot n^{O(1)}$, where $n$ is the number of vertices of the input graph and $t$ is its treewidth.
\end{theorem}

\begin{proof}
Let us prove the statement for \probkWMHE, and then explain how the proof extends for \probkWMHV.
Let $(G,c,w,\ell)$ be an instance of \probkWMHE, let $(\{X_i \mid i \in I\}, T=(I,F))$ be a nice tree decomposition of $G$ of width $t$, and let $r$ be the root of $T$.
Moreover, denote by $G_i$ the subgraph of $G$ induced by $\bigcup_j X_j$ where $j$ belongs to the subtree of $T$ rooted at $i$.

For every node $i$ of $T$ we set up a table $K_i$ indexed by all possible extended full $k$-colorings of $X_i$.
Intuitively, an entry of $K_i$ indexed by $f : X_i \to [k]$ gives the total weight of edges happy in $G_i$ under $f$.
It holds that an optimal solution is given by $\max_f \{ K_r[f] \}$.
In what follows, we detail the construction of the tables $K_i$ for every node $i$.
The algorithm processes the nodes of $T$ in a post-order manner, so when processing $i$, a table has been computed for all children of $i$.
\begin{itemize}
\item \textbf{Leaf node.} Let $i$ be a leaf node and $X_i = \{ v \}$. Obviously, $G_i$ is edge-free, so we have $K_i[f] = 0$.
As $k$ is fixed, $K_i$ is computed in constant time.
\item \textbf{Introduce node.} Let $i$ be an introduce node with child $j$ such that $X_i = X_j \cup \{ v \}$.
Put differently, $G_i$ is formed from $G_j$ by adding $v$ and a number of edges from $v$ to vertices in $X_j$.
The properties of a tree decomposition guarantee that $v \notin V(G_j)$, and that $v$ is not adjacent to a vertex in $V(G_j) \setminus X_j$.
It is not difficult to see that we set $K_i[f] = K_i[f \vert_{X_j}] + \sum_{p \in N_h(v)} w(pv)$, where $N_h(v)$ denotes the neighbors of $v$ colored with the same color as $v$.
It follows $K_i$ can be computed in time $O(k^{t+1})$.
\item \textbf{Forget node.} Let $i$ be a forget node with child $j$ such that $X_i = X_j \setminus \{ v \}$.
Observe that the graphs $G_i$ and $G_j$ are the same.
Thus, we set $K_i[f]$ to the maximum of $K_j[f']$ where $f' \vert_{X_i} = f$.
Since there are at most $k$ such colorings $f'$ for each $f$, we compute $K_i$ in time $O(k^{t+2})$.
\item \textbf{Join node.} Let $i$ be a join node with children $j_1$ and $j_2$ such that $X_i = X_{j_1} = X_{j_2}$.
The properties of a tree decomposition guarantee that $V(G_{j_1}) \cap V(G_{j_2}) = X_i$, and that no vertex in $V(G_{j_1}) \setminus X_i$ is adjacent to a vertex in $V(G_{j_2}) \setminus X_i$.
Thus, we add together weights of happy edges that appear in $G_{j_1}$ and $G_{j_2}$, while subtracting a term guaranteeing we do not add weights of edges that are happy in both subgraphs.
Indeed, we set $K_i[f] = K_{j_1}[f] + K_{j_2}[f] - q$, where $q$ is the total weight of the edges made happy under $f$ in~$X_i$.
The table $K_i[f]$ can also be computed in time $O(k^{t+2})$.
\end{itemize}
To summarize, each table $K_i$ has size bounded by $k^{t+1}$.
Moreover, as each table is computed in $O(k^{t+2})$ time, the algorithm runs in $k^t \cdot n^{O(1)}$ time, which is what we wanted to show.

The proof is similar for \probkWMHV, but each table now stores the total weight of the happy vertices under an extended full $k$-coloring.
In addition, for each vertex $v$ in a bag, we also store a bit indicating whether or not all forgotten neighbors of $v$ share its color, i.e., whether $v$ can still be made happy.
\end{proof}

The polynomial-time solvability of a problem on bounded treewidth graphs implies the existence of a polynomial-time algorithm also for other structural parameters that are polynomially upper-bounded in treewidth.
For instance, one such parameter is the \emph{vertex cover number}, i.e., the size of a smallest vertex cover that a graph has.
However, graphs with bounded vertex cover number are highly restricted, and it is natural to look for less restricting parameters that generalize vertex cover (like treewidth).
Another parameter generalizing vertex cover is \emph{neighborhood diversity}, introduced by Lampis~\cite{Lampis2012}.
Let us first define the parameter, and then discuss its connection to both vertex cover and treewidth.

\begin{definition}
In an undirected graph $G$, two vertices $u$ and $v$ have the same type
if and only if $N(u) \setminus \{ v\} = N(v) \setminus \{ u \}$.
\end{definition}

\begin{definition}[Neighborhood diversity~\cite{Lampis2012}]
A graph $G$ has neighborhood diversity $t$ if there exists a partition of $V(G)$ into $t$ sets $P_1,P_2,\ldots,P_t$ such that all the vertices in each set have the same type. Such a partition is called a type partition. Moreover, it can be computed in linear time.
\end{definition}
Note that all the vertices in $P_i$ for every $i \in [t]$ have the same neighborhood in~$G$. Moreover, each $P_i$ either forms a clique or an independent set in~$G$.

Neighborhood diversity can be viewed as representing the simplest of dense graphs.
If a graph has vertex cover number $d$, then the neighborhood diversity of the graph is not more than $2^d + d$ (for a proof, see~\cite{Lampis2012}).
Hence, graphs with bounded vertex cover number also have bounded neighborhood diversity.
However, the opposite is not true since complete graphs have neighborhood diversity~1.
Paths and complete graphs also show that neighborhood diversity is incomparable with treewidth.
In general, some NP-hard problems (some of which remain hard for treewidth), are rendered tractable for bounded neighborhood diversity (see e.g.,~\cite{Ganian2012,Gargano2015,Fiala2016}).

We proceed to present algorithms for \probkMHE and \probkMHV for graphs of bounded neighborhood diversity.
Consider a type partition of a graph $G$ with $t$ sets, and an instance of $I = (G,c,\ell)$ of \probkMHE.
If a set contains both precolored and uncolored vertices, we split the set into two sets: one containing precisely the precolored vertices and the other precisely the uncolored vertices.
After splitting each set, the number of sets is at most~$2t$.
For convenience, we say a set is \emph{uncolored} if each vertex in it is uncolored; otherwise the set is \emph{precolored}.
Let the uncolored sets be $P_1, P_2, \ldots, P_t$.
In what follows, we discuss how vertices in these sets are colored in an optimal solution.
We say a set is \emph{monochromatic} if all of its vertices have the same color.

\begin{figure}[t]
\centering
\includegraphics[scale=1.25]{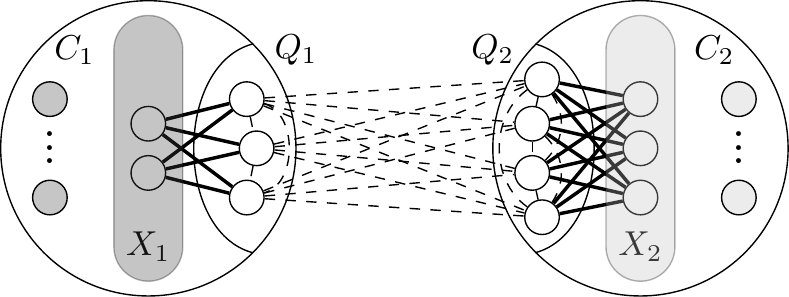}
\caption{A set of a type partition, where each vertex in $Q_1 \cup Q_2$ has the same type. The dashed edges appear exactly when $Q_1 \cup Q_2$ induces a clique. The set $Q_1$ forms a complete bipartite graph with both $X_1$ and $X_2$; likewise for $Q_2$ (edges omitted for brevity).
}
\label{fig:nd}
\end{figure}

\begin{lemma}
\label{lem:ndkmhe}
There is an optimal extended full coloring for an instance~$I$ of \probkMHE such that each uncolored set $P_i$ for $1 \leq i \leq t$ is monochromatic.
\end{lemma}
\begin{proof}
Consider any optimal extended full coloring for an instance $I$.
Suppose the vertices in a set $P_i$ belong to more than one color class.
Let $Q_1$ and $Q_2$ be the (disjoint and non-empty) sets of vertices of $P_i$ belonging to color classes $C_1$ and $C_2$, respectively.
Let $X_1$ and $X_2$ be the neighbors of the vertices in
$Q_1$ and $Q_2$ in color classes $C_1$ and $C_2$, respectively, as shown in Figure~\ref{fig:nd}.
Without loss of generality, let us assume that $|X_1| \leq |X_2|$.
By recoloring vertices in $Q_1$ with the color of $C_2$, we retain an optimal solution without disturbing the colors of other vertices.
If $|E(Q_1, Q_2)|$ is the number of edges between $Q_1$ and $Q_2$,
the gain in the number of happy edges by recoloring $Q_1$ is $|E(Q_1, Q_2)| + |Q_1|(|X_2|-|X_1|)$,
which is strictly positive if $P_i$ is a clique and non-negative if $P_i$ is an independent set.

In conclusion, we have shown that every optimal extended full coloring makes each $P_i$ inducing a clique monochromatic.
Moreover, there is an optimal extended full coloring making each $P_i$ inducing an independent set monochromatic.
\end{proof}

\begin{theorem}
\label{thm:ndkmhe}
For any $k \geq 1$, \probkMHE can be solved in time $O^{*}(2^t)$, where $t$ is the neighborhood diversity of the input graph.
\end{theorem}
\begin{proof}
First we construct a weighted graph $H$ from $G$ as follows:
merge each uncolored set into a single vertex.
Within a precolored set (i.e., a set that is not uncolored), merge vertices of the same color.
This merging operation may create parallel edges and self-loops in $H$.
Discard all self-loops in $H$.
Replace all parallel edges with a single weighted edge with weight equivalent to the number edges between the corresponding vertices.
Edges between the vertices in $G$ that are merged to the same vertex are treated as happy, as there is an optimal extended full coloring where the merged vertices are colored the same by Lemma~\ref{lem:ndkmhe}.
Clearly, $H$ has at most $t + kt$ vertices in which $t$ vertices are uncolored.

Now, \probkMHE on $G$ is converted to an instance of \probkWMHE on $H$.
By using Theorem~\ref{thm:exact-algs}, we can solve the instance of \probkMHE on $G$ in time $O^*(2^t)$.
\end{proof}

Using arguments similar to Lemma~\ref{lem:ndkmhe} we can state the following lemma.
This time, let $I = (G,c,\ell)$ be an instance of \probkMHV.

\begin{lemma}
\label{lem:ndkmhv}
There is an optimal extended full coloring for an instance~$I$ of \probkMHV such that each uncolored set $P_i$ for $1 \leq i \leq t$ is monochromatic.
\end{lemma}

\begin{proof}
Consider any optimal extended full coloring for an instance $I$.
Suppose the vertices in a set $P_i$ belong to more than one color class.
Let $Q_1$ and $Q_2$ be the (disjoint and non-empty) sets of vertices of $P_i$ belonging to color classes $C_1$ and $C_2$, respectively.
Let $X_1$ and $X_2$ be the neighbors of the vertices in
$Q_1$ and $Q_2$ in color classes $C_1$ and $C_2$, respectively, as shown in Figure~\ref{fig:nd}.
Without loss of generality, let us assume that $|X_1| \leq |X_2|$.
By recoloring $Q_1$ with the color of $C_2$, we get an optimal solution as well without disturbing the
colors of other vertices. The gain in the number of happy vertices by recoloring $Q_1$ is at most $|X_2| \geq 0$.
This proves that there is an optimal extended full coloring where each set $P_i$ is monochromatic.
\end{proof}

Using a construction similar to Theorem~\ref{thm:ndkmhe} we prove the following theorem.

\begin{theorem}
For any $k \geq 1$, \probkMHV can be solved in time $O^{*}(2^t)$, where $t$ is the neighborhood diversity of the input graph.
\label{thm:ndkmhv}
\end{theorem}

\begin{proof}
First we construct a weighted graph $H$ from $G$ as follows:
merge each uncolored set into a single vertex.
Within a precolored set, merge vertices of same color.
We assign a weight for each vertex of $H$ equivalent to the
number of vertices in $G$ that are merged to the vertex in $H$.
To make $H$ simple, discard all parallel edges and self-loops.
Clearly, $H$ has at most $t + kt$ vertices in which~$t$ vertices are uncolored.

Now, \probkMHV on $G$ is converted to an instance of \probkWMHV on~$H$.
By using Theorem~\ref{thm:exact-algs}, we can solve the \probkMHV on~$G$
in time $O^*(2^t)$.
\end{proof}

\section{Conclusions}
\label{sec:conclusions}
We further studied the algorithmic aspects of homophily in networks.
As explained, the positive results for \probkWMHE also imply tractability results for \probMultiwayCut, a problem studied by Langberg et al.~\cite{mwuc2006}.
Furthermore, our work invites for a more systematic study of the complexity of happy coloring for various structural parameters.
From a parameterized perspective, an obvious open question is whether \probkWMHV (or even its unweighted variant) admits a polynomial kernel.
We believe the answer is positive, and leave this for further work.

\bibliographystyle{splncs}
\bibliography{BibFile}

\end{document}